\documentclass[submission,copyright,creativecommons]{eptcs}
\providecommand{\event}{DCM 2015} 

\usepackage{breakurl}             
\usepackage{epsfig}
\usepackage{latexsym}
\usepackage{microtype}
\usepackage{hyperref}

\usepackage{times}
\usepackage{amssymb}
\usepackage{amsmath}
\usepackage{latexsym}
\usepackage{graphicx}

\usepackage[strings]{underscore} 
\usepackage{mdwlist}
\usepackage{subfigure}
\usepackage{wasysym}
\usepackage{amsthm } 
\usepackage{algorithm}
\usepackage{algpseudocode}

\selectfont 

\long\def\symbolfootnote[#1]#2{\begingroup%
\def\thefootnote{\fnsymbol{footnote}}\footnote[#1]{#2}\endgroup}

\newcommand{\oen}{[\hspace*{-1.2pt}[}
\newcommand{\cen}{]\hspace*{-1.2pt}]}

\newcommand{\prism}{\textsc{prism}}
\newcommand{\palps}{\textsc{palps}}
\newcommand{\spalps}{\textsc{s-palps}}
\newcommand{\wsccs}{\textsc{wsccs}}

\newcommand{\Nb}{{\textbf{Nb}}}
\renewcommand{\S}{{\textbf{S}}}
\renewcommand{\L}{{\textbf{Loc}}}
\newcommand{\Act}{{\textbf{Act}}}
\newcommand{\Proc}{{\textbf{Proc}}}
\newcommand{\Expr}{{\textbf{Expr}}}

\newcommand{\s}{{\mathbf{s}}}

\newcommand{\byp}[1]{\stackrel{ #1 }{{\longrightarrow}_p}}

\newcommand{\by}[1]{\stackrel{ #1 }{\longrightarrow}}
\newcommand{\C}{{\textbf{Ch}}}
\newcommand{\myloc}{{\sf myloc}}

\newcommand{\nil}{{\bf 0}}

\newcommand{\timed}{{\sf timed}}
\newcommand{\prob}{{\sf prob}}
\newcommand{\loc}[1]{{:}\langle{ #1 }\rangle}

\newcommand{\hide}{\backslash}
\newcommand{\calC}{{\mathcal C}}
\newcommand{\eqdef}{\stackrel{{\rm def}}{=}}
\newcommand{\ol}{\overline}

\renewcommand{\a}{\alpha}

\newcommand{\psumi}{{\bullet\hspace{-0.115in}{\sum_{i\in I}}}}
\newcommand{\psumj}{{\bullet\hspace{-0.115in}{\sum_{j\in J}}}}
\newcommand{\psum}{{\bullet\hspace{-0.1in}{\sum}}}

\newcommand{\psumn}{{\bullet\hspace{-0.31in}{\sum_{\ell\in \Nb(\myloc)}}}}

\newcommand{\comment}[1]{}

\newtheorem{definition}{Definition}
\newtheorem{theorem}{Theorem}

\newtheorem{example}{Example}

\begin{document}

\title
{Mean-Field Semantics for a Process Calculus for Spatially-Explicit Ecological Models}

\author{Mauricio Toro \institute{Universidad Eafit,  Colombia
    \email{mtorobe@eafit.edu.co}}
\and Anna Philippou \institute{University of Cyprus, Cyprus
    \email{annap@cs.ucy.ac.cy}}
\and Sair Arboleda \institute{Universidad de Antioquia, Colombia
   \email{sairorieta@yahoo.es}}
\and Mar\'{i}a Puerta  \institute{Universidad Eafit,  Colombia
    \email{mpuerta@eafit.edu.co }}
\and Carlos M. V\'{e}lez S. \institute{Universidad Eafit,  Colombia
    \email{cmvelez@eafit.edu.co }}
}

\def\titlerunning{Mean-Field Semantics for a Process Calculus for Ecological Models}
\def\authorrunning{M. Toro, A. Philippou, S. Arboleda, M. Puerta \& C. V\'{e}lez}

\maketitle

\pagestyle{plain}


\date{}

\begin{abstract}
We define a mean-field semantics for \spalps{}, a process calculus  for
spatially-explicit, individual-based modeling of ecological systems. The new
semantics of \spalps{} allows an interpretation
of the average behavior of a system as a set of recurrence equations.
Recurrence equations are a  useful approximation when dealing with a large
number of individuals, as it is the case in epidemiological studies.
As a case study, we compute a set of recurrence equations capturing the dynamics of an individual-based model
of the transmission of dengue in Bello (Antioquia), Colombia.
\end{abstract}

\section{Introduction}
The collective evolution of a group of individuals is of importance
in many fields; for instance, in system biology, ecology and epidemiology.
When modeling such systems, we want to know the emergent behaviors of the
whole population given a description of the low-level interactions of the individuals
in the system. As an example, in eco-epidemiology the focus is on the number of individuals
infected in a certain population  and how a small number of individuals infected may lead to an epidemic.

Eco-epidemiology can be seen as a particular case of \emph{population ecology}.
The main aim  of population ecology is to gain a better
understanding of population dynamics and make predictions about how
populations will evolve and how they will respond to
specific management schemes. In eco-epidemiology, such management schemes can be
 a cure to a disease, mechanisms to prevent a disease such as vaccines, or
mechanisms to prevent the \emph{vector} (species infected with a disease) to spread a disease.
To achieve these goals, scientists may construct models of ecosystems and management schemes (e.g., \cite{Yang2008}).

Various formalisms have been proposed in
the literature for the individual-based modeling of biological and ecological systems.
%
%
Examples of such formalisms include the calculus of looping
sequences~\cite{BarbutiMMT06} and its spatial extension
\cite{BMMP09}, cellular automata~\cite{FM04,CYL09}, Petri nets~\cite{KleijnKR11},
synchronous automata~\cite{DrabikMM11}, P systems~\cite{BCPM06,abs-1008-3301, PinnaS08}
and process calculi (or process algebras)~\cite{SBB01,MNS08}.

In our work,
we are interested in the application of process calculi for studying the
population dynamics of ecological systems. Process calculi are formal frameworks
to model and reason about concurrent systems and provide constructs to express
sequential and parallel composition of processes, as well as different means of communication
between processes.
In contrast to
the traditional approach to modeling ecological systems using ordinary
differential equations which describe a system in terms of changes in the population as a whole,
process calculi are suited towards the so-called ``individual-based''
 modeling of populations. Process calculi enable one to
describe the evolution of each individual of the population as a process and, subsequently, to compose a
set of individuals (as well as their environment) into a complete ecological system.
 Process calculi include features such as time \cite{timedpi, pntcc},
probability \cite{Tofts94} and stochastic behavior \cite{HillstonTG12}.
Furthermore, following a model construction, one can use model-checking tools for automatically analyzing
properties of the models (e.g., \cite{MNS10, TPSK14}) as opposed to just simulating trajectories, as it is typically
carried out in most ecological studies.

In a previous work, we presented \palps{} (Process Algebra with Locations for Population Systems), a process calculus
developed for modeling and reasoning about spatially-explicit individual-based systems~\cite{PTA13}.
In \palps{}, individuals
are modeled using discrete time and probabilistic behavior, and space is modeled as a graph of discrete locations.
We associated \palps{} with a translation to the probabilistic model checker \prism{} \cite{prism}
to perform more advanced analysis of ecological models.
Our experiments with  \palps{} ~\cite{PTA13,PT13} delivered promising results via the use of statistical
model checking provided by \prism{}. Nonetheless, the results also revealed a
limitation:  the interleaving nature of the parallel composition operator led to a high
level of nondeterminism and, thus, a very quick explosion of the state space. Moreover, the 
interleaving nature of \palps{} 
comes in contrast to the usual approach of modeling adopted by ecologists
where it is considered that execution evolves in phases during which
individuals of a population engage \emph{simultaneously} in some process, such as birth, dispersal and
reproduction.

To alleviate the problem of the interleaving nature of parallel composition and the
high degree of nondeterminism in \palps{}, we proposed  a new semantics of \palps{},
which captures more faithfully the synchronous evolution of populations and
removes as much unnecessary nondeterminism as possible. Our proposal
consisted of a synchronous extension of \palps{}, named \emph{synchronous \palps{}} (\spalps{}) \cite{TPSK14}.
The semantics of \spalps{} implements
the concept of \emph{maximum parallelism}: at any given time all individuals
that may execute an action will do so simultaneously~\cite{TPSK14}. Furthermore, we
proposed a new translation of \spalps{} to \prism{} which implemented this synchronous
semantics, as well as other features that removed the restrictions existing in the original
framework. This led to a significant improvement regarding the size of \spalps{} models that can
by analyzed via translation to \prism{} in the range of hundreds of individuals. However, in epidemiological systems, components can number in millions.

To deal with this challenge, in this paper, we present a \emph{mean-field
semantics} to represent the average behavior of \spalps{} systems for populations
of potentially millions of individuals. Mean-field semantics gives a deterministic
approximation of the average behavior of a system, given low-level specifications at
the individual level in terms of discrete-time and discrete-space mean-field equations.
In this work we propose an algorithm of polynomial-time complexity for producing
the mean-field equations given an \spalps{} model of a system. The algorithm avoids computing the complete state-space
of a system and its complexity is independent of the size of the populations.
Rather, given the stochastic nature of the systems in question, the accuracy
of the method relies on the fact that the numbers of each agent in the system are sufficiently large.
We illustrate the application of our semantics for the construction of mean-field equations
of an \spalps{} model of the transmission of dengue in Bello, Colombia.

Mean-field semantics have been proposed for a number of process calculi
including \textsc{pepa} \cite{HillstonTG12,Tribastone12} and \wsccs{} \cite{Tofts94,MNS10}.
The former line of work differs from our semantics since the underlying model is
continuous time. Instead, our work is closely related to that of~\cite{MNS10} for
\wsccs{}. Our semantics extends that of~\cite{MNS10} in two ways. First, we extend
the semantics of \spalps{} to deal with locations, since \spalps{}  includes an explicit notion of discrete space
not present in \wsccs. Second, the nature of our calculus and, specifically, the presence
of an explicit probabilistic operator, as opposed to weights, and the absence of nondeterminism
at the level of individuals, yields a simpler semantics, as well as the lifting of some of
the restrictions imposed in~\cite{MNS10}. As related work, we also mention the
mean-field semantics proposed for reactive networks in
\cite{reactiveNetworks} which, however, is not directly related with our aim of
providing this analysis capability to the spatially-explicit process calculus \spalps{}.


The remainder of the paper is as follows. In Section
\ref{sec:Calculus} we present the syntax and
the
semantics of \spalps{}. In Section~\ref{sec:MeanField} we present a mean-field
semantics for \spalps{}. We apply
our techniques to study the population dynamics of dengue  in Section \ref{sec:case-study}.
 Section \ref{sec:Conclusions} presents conclusions and future work.

\section{Synchronous \palps{}}
\label{sec:Calculus}

In \spalps{}, we consider a system as a set
of individuals operating in space, each belonging to a certain species
and inhabiting a location.
Individuals who reside at the
same location may communicate with each
other upon channels (e.g., for preying) or they
may migrate to a new location.
\spalps{}  models probabilistic events with the aid
of a probabilistic operator. 

The syntax of \spalps{} is based on the following basic entities: (1)
$\S$ is a set of species ranged over by $\s$, $\s'$. (2)
$\L$ is a set of locations ranged over by $\ell$, $\ell'$. The habitat
is then implemented via a relation $\Nb$, where $(\ell,\ell')\in \Nb$
exactly when  $\ell$ and $\ell'$ are neighbors. 
(3) $\C$ is a set of channels ranged over by lower-case strings.
 The syntax of \spalps{}  is given
at two levels, the individual
level ranged over by $P$ and  the system level ranged over by $S$ which are defined as follows:
\begin{eqnarray*}
  P &::=& \nil  \:\:|\:\: \eta.P \:\:|\:\: \psumi p_i{:}P_i \:\:|\:\: \gamma?\,(P_1,P_2)\:\:|\:\: 
            P_1| P_2\:\:|\:\: C
\\
  S &::=& \nil  \:\:\:|\:\:\: P\loc{\s, \ell, q} \:\:\:|\:\:\: S_1
  \,\|\,S_2 \:\:\:|\:\:\: S\hide L
\end{eqnarray*}
where $L\subseteq \C$, $I$ is an index set, $p_i\in(0,1]$ with
$\sum_{i\in I} p_i = 1$, $C$ ranges over a set of process constants
$\calC$, each with an associated definition of the form $C\eqdef P$,
and the actions that a process can perform are
\[
\eta ::=  a \:\:|\:\: \ol{a} \:\:|\:\:
           go\; \ell \:\:|\:\: \surd\,\;\;\;\;\;\;\;\;
\gamma  ::= a \:\:|\:\: \ol{a} \]

Beginning with the \emph{individual level}, $P$ can be one of the
following:

\begin{itemize*}
\item Process $\nil$ represents the inactive individual, that
is, an individual who has ceased to exist.

\item Process $\eta. P$
describes the action-prefixed process which executes action $\eta$
before proceeding as $P$. An activity $\eta$ can be an
input action on a channel $a$, written simply as $a$; an output action on a channel $a$, written as $\ol{a}$; a
movement action to location $\ell$, $go\, \ell$; or the tick
action $\surd$ that indicates a discrete-time unit on a global tick action, $\surd$. Actions of the form $a$, and $\ol{a}$,
$a\in\C$, are used to model activities performed by an
individual; for instance, preying and reproduction. 

\item Process $\psum_{i\in I} p_i{:} P_i$ 
represents the probabilistic choice between processes $P_i$, $i\in
I$. The process randomly selects an index $i\in I$ with probability
$p_i$, and then evolves to process $P_i$. We write $p_1{:}P_1 \oplus
p_2{:}P_2$ for the binary form of this operator.

\item Process $\gamma?\,(P_1,P_2)$ depends on the availability of a
communication on a certain channel as described by $\gamma$.
 If a communication is available according to $\gamma$ then the communication
 is executed and
the flow of control proceeds according to $P_1$. If not, the
process proceeds as $P_2$. This operator is a deterministic operator
as, in any scenario, the process $\gamma?\,(P_1,P_2)$ proceeds as either
$P_1$ or $P_2$ but not both, depending on the availability of the
complementary action of $\gamma$ in the environment in which the
process is running. 

\end{itemize*}

Moving on to the \emph{population level}, population systems are built by composing in parallel
sets of located individuals. A set of $q$ individuals of species $\s$
located at location $\ell$  is defined as
$P\loc{\s,\ell,q}$. In a composition
$S_1\| S_2$ the components may proceed while synchronizing on
their actions following a set of restrictions. These restrictions
enforce that probabilistic transitions take precedence over the
execution of other actions and that time proceeds synchronously in
all components of a system. That is, for $S_1\| S_2$ to execute
a $\surd$ action, both
$S_1$ and $S_2$ must be willing to execute $\surd$. Action $\surd$ measures a
tick on a global clock. These time steps are abstract in the sense
that they do not necessarily have a defined length and, in practice, $\surd$ is used to separate the rounds of
an individual's behavior. In the case of multi-species systems these actions
must be carefully placed in order to synchronize species with possibly different time scales.

System $S\hide L$ models the
restriction of channels in  $L$ within $S$.
This construct is important to define \emph{valid systems}: We define a
\emph{valid} system to be any process of the form $S\hide L$ where, for all of $S$'s  subprocesses
of the form $a?(P,Q)$ and $\ol{a}?(P,Q)$ we have that $a\in L$. Hereafter, we consider only
 valid systems.

\begin{example}
\label{MainExample}{\rm \ Let us consider
a species $\s$ where individuals cycle through a dispersal
phase followed by a
reproduction phase. Further, suppose that the habitat is a ring of size $m$ where
the neighbors of location $\ell$ are $\ell\pm 1$.
In \spalps{}, we may model $\s$ by $P_0$, where
\begin{eqnarray*}
P_0 &\eqdef &{\psumn}\frac{1}{2}:go\,\ell.\surd.P_1\quad\quad
P_1 \eqdef  p{:}\surd.(P_0 | P_0)\; \oplus
\;(1-p){:}\surd.(P_0 |P_0| P_0)
\end{eqnarray*}

According to the previous definition, during the dispersal phase, an
individual moves to a neighboring location which is chosen
probabilistically among the neighboring locations of the \emph{current
location} ($\myloc$) of the individual.  Subsequently, the flow of control
proceeds according to
 $P_1$ which models the probabilistic production
of one offspring (case of $P_0|P_0$) or two offspring (case of
$P_0|P_0|P_0$).  A system that contains two individuals at a
location $\ell$ and one at location $\ell'$ can be modeled as
\[\mathit{System} \eqdef P_0\loc{\s,\ell,2} |P_0\loc{\s,\ell',1}\,.\]
}
\end{example}





The semantics of \spalps{} is defined operationally via two transition relations, the \emph{non-deterministic
transition relation} and the \emph{probabilistic transition relation} yielding transition systems that
can be easily translated into \emph{Markov decision processes} \cite{Puterman1994}. The main
features of the semantics is that probabilistic transitions take precedence over all other actions
and that all processes must synchronize on timed actions. Finally, at any given time,
all individuals that may execute an action will do so simultaneously. A full account of the semantics
can be found in~\cite{TPSK14}.

\section{Mean-field semantics for \spalps{}}
\label{sec:MeanField}

Using the operational semantics of \spalps{}, we can study the transient
dynamics of a system: the time series evolution of the model. This is
necessary for the simulation of models that can be obtained by translating
\spalps{} to \prism{}. Using \prism{} it is also possible to use  model
checking or approximated model checking (i.e., statistical model checking).
Although this approach can be effective for \spalps{} models with fairly
large state spaces (consisting of populations in the range of a few hundreds
of individuals), the size of a  state space is exponential
in the number of components and locations, and, in epidemiological systems,
components can number in millions. To address this challenge, in this section
we develop a mean-field semantics for \spalps{} which can be used for reasoning
about systems with very large populations.
To compute the mean-field semantics,
we proceed in 3 steps: (1) compute the initial-state matrix, (2) compute the
state-transition table and (3) compute the mean-field equations. In particular,
we begin by assuming an  \spalps{} model of the form:
\[ \mathit{System} = (\Pi_{1\leq j \leq m}P_1\loc{\s_1,\ell_j,q_{1,j}} |  \ldots |\Pi_{1\leq j\leq m}P_n\loc{\s_n,\ell_j,q_{n,j}})\hide L\]
where $P_1,\ldots,P_n,$ is the set of all processes the populations may evolve into,
$\s_1,\ldots,\s_n\in\S $, $\ell_1,\ldots,\ell_m\in\L$ is the set of all locations
in the system,  and the $q_{i,j}\geq 0$ are the sizes of the population of individuals
at state $P_i$ at location $\ell_j$ where, if a location-state
pair $(P_i,\ell_j)$ for some species $\s$  is not present in the initial configuration then $\mathit{System}$
includes the component $P_i\loc{\s,\ell_j,0}$.

Restrictions for our method are the following: as usual, the numbers of the agents $q_i$ must
be sufficiently large and process constants must be guarded. That is, we do not allow
definitions of the form $C \eqdef P|C$, since these yield infinite-sized systems.



\paragraph{1. Initial-state matrix ($\mathit{Init}$). } This matrix, $\mathit{Init}$, captures the
initial configuration of the system under study by noting the number of individuals
of each type at each location. It is a matrix of size $n\times m$,
where $n$ is the number of all accessible process-states and $m$ the number of distinct locations in
the system
(see the definition of $\mathit{System}$ above) and it is obtained directly
from the definition of $\mathit{System}$. Specifically, $\mathit{Init}[i,j]=q$ where
$q$ is the number of individuals of state $P_i$ at location $\ell_j$.

\paragraph{2. State-transition table ($\mathit{STT}$). } This $3$-dimensional table, $\mathit{STT}$, shows how processes evolve from
one state to another and how their locations change. Each entry in the state-transition table is an expression that
captures the average evolution, after the execution of a single action of a
process $P$ at a location $\ell$ to some process $Q$ at a location $\ell'$.
This is expressed as a function of the size of the population of process $P$.
 Formally, matrix
$\mathit{STT}$ is of size $n\times n\times m$. We point out again that $n$ is the number of different states individuals may engage in and $m$ the number of locations in a system. Note that while these quantities may be big, if we are considering a detailed model of a system (many process states and large number of locations), typically, they are fairly small and, most importantly, they are independent of the size of the populations considered as well as the size of the system under consideration.
 This matrix
captures the evolution after one action step and not necessarily after a time unit,
since actions under study 
may include actions such as
$go\;\ell$, communication actions and probabilistic actions. 
The entries of the matrix are expressions that capture the number of processes of
a certain type that have evolved at
a location as a function of the number of various processes at different locations
that may evolve into the specific state-location pair, in the previous step of the
system.

In order to capture
the evolution in a manner compatible with the original \spalps{} semantics, we employ the following notions:
 $\timed(S)$ captures whether $S$ may engage in a timed action
    (all its active components may execute $\surd$);
 $\prob(S)$ captures whether $S$ may engage in a probabilistic actions (at least
     one of its components may execute a probabilistic action).
These notions are essential to capture that probabilistic actions
take precedence over all other actions and that $\surd$ actions may take place
only if all components of the system are willing to synchronize on a timed step.
In particular, given a system $S$, to construct the values of the transition matrix $\mathit{STT}$
capturing the evolution of its components $P_i\loc{\s_i,\ell_i,q_i}$ we use a function
$\oen\cdot\cen:\Proc\times\Act \rightarrow \mathcal{P}(\Expr:\Proc)$, where $\Proc$
is the set of all processes of the form $P\loc{\s,\ell,q}$ and $\Expr$ is an expression
capturing the evolution in question. In particular, given a process  $P\loc{\s,\ell,q}$
and an action  $\a$, $\oen P\loc{\s,\ell,q}, \a\cen$ returns a set of expressions
$e_i:P_i\loc{\s,\ell_i,q_i}$  capturing the set of processes in which $P\loc{\s,\ell,q}$ may
evolve and the concentration $e_i$ for each of these processes as a function of the concentration
of $P\loc{\s,\ell,q}$.
We proceed to define
this function.
We begin with the evolution according to probabilistic transitions and timed actions,
where we have:
\begin{eqnarray*}
\oen P\loc{\s,\ell,q}, a)\cen & = &\langle\rangle,\;\; \mbox { if } \prob(S)\\
\oen P\loc{\s,\ell,q}, prob \cen & = & \langle p_i\cdot P_t: P_i \loc{\s,\ell,q}|i\in I\rangle,\;\; \mbox {if } P=\psumi p_i{:}P_i\\
\oen P\loc{\s,\ell,q} , \surd \cen & = &\langle P_t: P'\loc{\s,\ell,q}\rangle,\;\; \mbox { if } P = \surd .P'  \mbox{ and } \timed(S)\\
\oen P\loc{\s,\ell,q} , \surd \cen & = &\langle \rangle,\;\; \mbox { if } P = \surd .P'  \mbox{ and } \lnot\timed(S)
\end{eqnarray*}

Thus, no communication on channel $a$ may take place if a process occurs within a system
satisfying $\prob(S)$. Similarly, a $\surd$ action may not take place if the process
does not occur within a system satisfying $\timed(S)$. On the other hand,
probabilistic transitions may take place freely and so do $\surd$ actions within
timed systems. Note that in the above, we write $P_t$ for the number of agents $P$
at step $t$.

Moving on to the execution of a movement action, we define:
\begin{eqnarray*}
\oen P\loc{\s,\ell,q} , \tau_{go,\ell} \cen & = &
        \langle P_t:P'\loc{\s,\ell',q}\rangle,\;\; \mbox { if } P = go\,\ell' .P'
\end{eqnarray*}
This leaves us with the execution of channel-based actions where we distinguish the following cases:

\begin{itemize}
\item
If $P = \eta.P'$, where $\eta\in\{a,\ol{a}\}$ and $a\not\in L$, that is $a$ does not belong to the set of restricted channels, then we have
\begin{eqnarray*}
\oen {P\loc{\s,\ell,q} , \eta} \cen & = &
        \langle P_t:P'\loc{\s,\ell,q}\rangle,\;\; \mbox { if } P = \eta .P'
\end{eqnarray*}

\item If  $P = \eta.P'$ where $\eta\in\{a,\ol{a}\}$ and $a\in L$, then
the number of agents evolving to $P'$ depends on the number of agents co-located with
$P$ and available to execute action $\eta$ and the complementary action $\ol{\eta}$.
Let us write $X_t$ for the number of co-located agents able to
execute $\eta$ and $Y_t$  for the number of co-located agents able to execute the complementary
action $\ol{\eta}$. If $Y_t\geq q + X_t$ then all agents of type $P$ will proceed to state
$P'$. If not, then the mean change in agent $P$ is expressed as
\[\frac{\sum_{k=1}^{q}k\dbinom{q}{k}\dbinom{X_t}{Y_t-k}}{\sum_{k=1}^{q}\dbinom{q}{k}\dbinom{X_t}{Y_t-k}}\]
This term can be simplified using \emph{Vandermonde's Convolution} and standard theory regarding the binomial coefficient to
$\frac{q\cdot Y_t}{X_t}$  \cite{Graham1994}. Thus, we have:
\begin{eqnarray*}
\oen {\eta.P\loc{\s,\ell,q} , \eta} \cen & = &
        \langle \min(q, \frac{q\cdot Y_t}{X_t}):P'\loc{\s,\ell,q}\rangle
\end{eqnarray*}

\item
Finally, we have to consider the evolution of a $P=\gamma?(P_1,P_2)$ process. In such
processes, we know that $\gamma\in\{a,\ol{a}\}$ where $a\in L$. Thus, the evolution
is similar to the previous case. The point in which this case differs is when there
is not a sufficient number of collaborating agents to provide the complementary
$\ol{\gamma}$ actions. In such a case, a number of instances of the
 process will evolve to $P_2$ thus, giving:
\begin{eqnarray*}
\oen {\gamma?(P_1,P_2)\loc{\s,\ell,q} , \eta} \cen & = &
        \langle \min(q, \frac{q\cdot Y_t}{X_t}):P_1\loc{\s,\ell,q},(q-\min(q, \frac{q\cdot Y_t}{X_t})):P_2\loc{\s,\ell,q}\rangle
\end{eqnarray*}
\end{itemize}
\paragraph{3. Mean-field equations (\textsc{mfe}s). } Using the state-transition
table and the initial-state matrix, we can derive a set of recurrence equations
 that represent the mean-field semantics of a  system.  The system of recurrence
equations contains one variable for each different process and at each location
in a system. A
variable $P_i(t)@\ell_j$ represents the mean number of individuals of process
$P_i\loc{\s,\ell_j,m'}$, at time $t$ and location $\ell_j$ and a variable $P(t-1)@\ell_j$
represents the number of individuals of process $P_i$ at location $\ell_j$ during
time $t-1$.  Formally, $P_i(t)@\ell_j$ is defined by

\begin{center}
$P_i(t)@\ell_j = \left\{
\begin{array}{c l}
 \mathit{Init}[i,j] & t = 0 \\
 \sum_{1\leq k \leq n} \mathit{STT}[k][i][\ell_j] & \text{otherwise}
\end{array}
\right.
$
\end{center}

In the first case, for $t=0$, the value is obtained from the initial-state
matrix. The second case, for $t > 0$, the value is obtained from state-transition
from all the processes  in the system. According to the state-transition table,
processes that do not derive into $P'\loc{\s,\ell,m'}$ are said to produce $0$
individuals of process $P'\loc{\s,\ell,m'}$ in the next time unit.

 Algorithm \ref{MFEalgorithm} is the pseudocode to construct the state-transition table $\mathit{STT}$. For simplicity, we assume that the entries
 of the state-transition table are expressions over the number of individuals of each process.

\begin{algorithm}
\caption{Algorithm to compute the state-transition table (\textsc{stt}) of a \spalps{} model}
\label{MFEalgorithm}
\begin{algorithmic}[1]
\Procedure{Compute\_\textsc{stt}} {\textit{System}}
\State $\mathit{STT}$ $=$ matrix of dimension $n \times n \times m$ initialized with $0$


  \For{each $P_i, \ell_j \in \textit{System}$ }
  \If{ $P\texttt{==}\ \ \psumj p_j{:}P_j$} $\mathit{STT}[i][j][\ell_j] \texttt{+=}   p_j\cdot P_i@(t-1), \forall j\in J$
  \ElsIf{ $P_i \texttt{==} \surd.P_j$ and $\timed(\textit{System})$} $\mathit{STT}[i][j][\ell_j] \texttt{+=}  P_i@(t-1)$


  \ElsIf{$ P_i \texttt{==} go\,\ell_k .P_j$ } $\mathit{STT}[i][j][\ell_k] \texttt{+=}  P_i@(t-1)$

  \ElsIf{$\eta \texttt{==} a$ or $\eta \texttt{==} \ol{a}$, $a\not\in L$ and $P_i \texttt{==} \eta.P_j$} $\mathit{STT}[i][j][\ell_j] \texttt{+=} P_i@(t-1)$
  \ElsIf {$\eta \texttt{==} a$ or $\eta \texttt{==} \ol{a}$, $a\in L$ and $P_i \texttt{==} \eta.P_j$}
        \State $X_{t-1} = $ number of co-located agents executing $\eta$
        \State $Y_{t-1} = $ number of co-located agents executing $\ol{\eta}$
        \State $\mathit{STT}[i][j][\ell_j] \texttt{+=} \min(P_i@(t-1), \frac{P_i@{t-1}\cdot Y_{t-1}}{X_{t-1}})$
        \State $\mathit{STT}[i][i][\ell_j] \texttt{+=} P_i@(t-1)- \frac{P_i@{t-1}\cdot Y_{t-1}}{X_{t-1}})$
  \ElsIf {$\eta \in \{a, \ol{a}\}$ and $a \in L$ and $P_i \texttt{==} \gamma?(P_j,P_k)$}
        \State $X_{t-1} = $ number of co-located agents executing $\eta$
        \State $Y_{t-1} = $ number of co-located agents executing $\ol{\eta}$
        \State $\mathit{STT}[i][j][\ell_j] \texttt{+=} \min(P_i@(t-1), \frac{P_i@{t-1}\cdot Y_{t-1}}{X_{t-1}})$
        \State $\mathit{STT}[i][k][\ell_j] \texttt{+=} P_i@(t-1)- \frac{P_i@{t-1}\cdot Y_{t-1}}{X_{t-1}})$
  \EndIf
\EndFor
\EndProcedure
\end{algorithmic}
\end{algorithm}

\begin{theorem}
 The complexity of Algorithm \ref{MFEalgorithm}
 is $O(n^2\cdot m)$ where $n$ is the number of different  states individuals may engage in and $m$ the number of different locations in the system.
\end{theorem}

\begin{proof}
The algorithm consists of a loop that fills in the positions of the $3$-dimensional matrix
$\mathit{STT}$. Assuming that by some preprocessing we have the set of processes executing
the actions $\eta$ and $\ol{\eta}$ of interest, each value of the matrix can be computed at
constant time, yielding the result. 


\end{proof}



\begin{example}
\label{MainExampleMF}\rm \ In what follows we present the mean-field semantics of Example \ref{MainExample}.
In this example, we have a total of $7$ process states: $R_1 = P_0$, $R_2 = go\,(\myloc+1).R_4$, $R_3 = go\,(\myloc-1).R_4$,
$R_4 = \surd.P_1$, $R_5 = P_1$,
$R_6 = \surd.(P_0 | P_0)$ and $R_7 = \surd.(P_0 |P_0| P_0)$.
Furthermore, let us assume that the habitat is a ring of size $4$.

\begin{enumerate}
\item  In the initial-state, there are
$2$ individuals of process $P_0$ at location $1$ and $1$ individual of process
$P_0$ at location $2$. This is represented in the following $4\times 7$ initial-state matrix:
\[[[2, 0, 0, 0, 0, 0,0],[1,0,0,0,0,0,0],[0,0,0,0,0,0,0],[0,0,0,0,0,0,0]]\]

\item Using the methodology defined above, the transition matrix for the example is given below.
Note that in fact this is a $3$-dimensional matrix, the third dimension
being the location dimension. To capture this in two dimensions we write $@\ell'$ whenever
the resulting individuals have moved to another location, $\ell'$ being this location.

\begin{center}
\begin{tabular}{|c | c | c | c | c | c | c | c|c|}
\hline
&  & $R_1$ & $R_2$ & $R_3$ & $R_4$ & $R_5$ & $R_6$ & $R_7$ \\
\hline
$R_1\loc{\s,\ell,q}$ & $prob$ & &$\frac{1}{2}\cdot q$ & $\frac{1}{2}\cdot q$  & & & & \\
\hline
$R_2\loc{\s,\ell,q}$ & $go$ & & & & $q@(\ell+1)$  & & & \\
\hline
$R_3\loc{\s,\ell,q}$ & $go$ & & & & $q@(\ell-1)$  & & & \\
\hline
$R_4\loc{\s,\ell,q}$ & $\surd$ & & & & & $q$ & &  \\
\hline
$R_5\loc{\s,\ell,q}$ & $prob$ & & & & & & $p\cdot q$ & $(1-p)\cdot q$ \\
\hline
$R_6\loc{\s,\ell,q}$ & $\surd $ & $2\cdot q$ & & & & & &  \\
\hline
$R_7\loc{\s,\ell,q}$ & $\surd $ & $3\cdot q$ & & & & & &  \\
\hline
\end{tabular}
\end{center}
\item Let us write $P(t)@\ell$ for the average number of individuals of
 process $R$, at time $t$ and location $\ell$. This can be computed as follows:
 $R_i(t)@\ell $ is equal to $2$ if $i,t,\ell=1,0,1$, equal to $1$ if $i,t,\ell=1,0,2$, and $0$, if $t=0$ and $i\neq 1$ or $\ell\not\in\{1,2\}$,
 whereas if $t>0$ we have:
 \begin{eqnarray*}
 R_1(t)@\ell & = & 2\cdot R_6(t-1)@\ell + 3\cdot R_7(t-1)@\ell\\
 R_2(t)@\ell & = & \frac{1}{2}\cdot R_1(t-1)@\ell\\
 R_3(t)@\ell & = & \frac{1}{2}\cdot R_1(t-1)@\ell\\
 R_4(t)@\ell & = & R_2(t-1)@(\ell+1) + R_3(t-1)@(\ell-1)\\
 R_5(t)@\ell & = & R_4(t-1)@\ell\\
 R_6(t)@\ell & = & p\cdot R_5(t-1)@\ell\\
 R_7(t)@\ell & = & (1-p)\cdot R_5(t-1)@\ell
\end{eqnarray*}

By manipulating the equations and restricting attention to how the system
evolves between $\surd$ actions, we obtain:
\begin{eqnarray*}
R_1(t)@\ell & = & 2\cdot p\cdot  R_5(t-2)@\ell + 3\cdot(1-p)\cdot R_5(t-2)@\ell\\
R_5(t)@\ell & = & \frac{1}{2}\cdot R_1(t-3)@(\ell+1)+\frac{1}{2}\cdot R_1(t-3)@(\ell-1)
\end{eqnarray*}

\end{enumerate}
\end{example}

\section{Correctness of the mean-field semantics}
\label{correctness}

In this section, we prove the correctness of our mean-field semantics by
establishing the relation between the derived mean-field equations and
the \spalps{} semantics. To achieve this, we first define an encoding
of the operational semantics of \spalps{} into a \emph{discrete-time Markov chain}
(\textsc{dtmc}). Then we show that the recurrence equations obtained from
the mean-field semantics of \spalps{} are equivalent to the recurrence equations
obtained from \textsc{dtmc}-semantics of \spalps{}.
This follows a result of~\cite{Kurtz1970}, according to which it is possible to derive
\emph{ordinary differential equations (ODEs)} as an approximation of the average behavior of a \textsc{dtmc}. At the limit,
where the \textsc{dtmc} consists of infinitely many agents, the mean of the \textsc{dtmc}
is equivalent to the derived ODE's.

\paragraph{Derivation of a \textsc{dtmc} from a \spalps{} system.}

The semantics of \spalps{} is given operationally via a structural operational semantics~\cite{TPSK14}.
This semantics is given in terms of two transition relations, a non-deterministic
transition relation and a probabilistic transition relation, which give rise
to labeled transition systems that present both non-deterministic and probabilistic states.
In the context of this work, we present a method for interpreting such transition systems as
a \textsc{dtmc}, under an abstract bisimulation. Essentially, an
abstract bisimulation is an equivalence relation that allows us to disregard the structure
of the non-deterministic choices and just
look at the probabilities of reaching any particular state. This approach to derive a
\textsc{dtmc} from a probabilistic process calculus, was proposed first by Tofts for process
calculus \wsccs{} \cite{Tofts94}.

To begin with, we define an abstract notion of evolution as follows:
\[P\by{\beta[p]}P' \mbox{ if an only if }
    P\byp{p_1}P_1\byp{p_2}\ldots \byp{p_{n-1}}P_{n-1}\byp{p_n}P_n\by{\beta}P'
    \mbox{ where } p = \prod_{1\leq i \leq n} p_i\]
We lift this notion to evolution into a set of processes as follows, where $S$ is a set of \spalps{} processes:
\[P\by{\beta[p]}S \mbox{ if an only if }  p = \sum\{p_i\mid P_i\by{\beta[p_i]} P, P\in S\}\]

We may now define the notion of abstract bisimulation as follows:
\begin{definition}\label{equivalenceRelation} {\rm \
An equivalence relation ${\cal R} \subseteq \Pr \times \Pr$ is an \emph{abstract bisimulation}
if $(P,Q) \in {\cal R}$ implies that for all equivalence classes $S \in \Pr/{\cal R}$, actions
$\beta$, and for all $p \in [0,1]$,
$P  \by{\beta[p]} S$ if and only  $Q  \by{\beta[p]} S$.

We say that two processes are \emph{abstract bisimulation equivalent}, written $P \sim Q$, if there
exists an abstract bisimulation ${\cal R}$ such that $(P,Q)\in{\cal R}$.

}
\end{definition}

This relation can be used to translate any \spalps{} system into a
\textsc{dtmc}: by building an abstract bisimulation on the set of states of
the system, we obtain a \textsc{dtmc} whose states are the equivalence classes of
the equivalences relation and the transition labels are $\beta[p]$,
as defined above.

We now turn to proving the relation between the \textsc{dtmc} semantics and the
mean-field semantics of \spalps{}. To do this we refer to~\cite{Kurtz1970} where limit theorems were
presented relating the mean of Continuous Time
Markov Chains and Discrete Time Markov Chains  to ordinary differential equations. 
In particular~\cite{Kurtz1970} shows that,
at the limit, where a DTMC consists of infinitely many agents, the mean of the Markov chain
is equivalent to a derived set of
ODEs. An intermediate step of Kurtz's proof produces terms equivalent to those 
of our mean-field semantics. We use this to show
the correctness of our semantics.
In particular, we refer to a result of~\cite{Kurtz1970} capturing the conditions under
which the limit theorem holds and then verify that these conditions apply to 
\spalps\ models. The theorem states the following:

\begin{theorem} \label{KurtzTheorem}{\rm\ \ 
Let $X_n(k)$ be a sequence of discrete time Markov processes with measurable state spaces
$(E_n, {\cal B}_n)$, $E_n\in {\cal B}^k$, the Borel sets in $\mathbb{R}^k$ and
one step transition functions denoted by
\[\mu_n(x,\Gamma) = P\{X_n(k+1) \in \Gamma | X_n(k) = x \}\]
Suppose there exist sequences of positive numbers $\alpha_n$ and $\epsilon_n$
\[\lim\limits_{x\to\infty} \alpha_n = \infty \mbox{ and }\lim_{x\to\infty} \epsilon_n = 0\]
such that
\[\sup_{n} \sup\limits_{x \in  E_n} \alpha_n \int\limits_{E_n} |z-x| \mu_n (x,dz) < \infty\]
and
\[\lim\limits_{x\to\infty} \sup\limits_{x \in E_n} \alpha_n \int\limits_{|z-x| > \epsilon_n} |z-x| \mu_n (x,dz) = 0\]
Then, for every $\delta > 0$, $t > 0$
\[\lim\limits_{n \to \infty} \sup\limits_{x \in E_n} P \{ \sup\limits_{k \leq \alpha_n t} | X_n(k) - X_n(0) - \sum_{l=0}^k \frac{1}{\alpha_n} F_n(X_n(l)) | > \delta \mbox{ where } X_n(0) = x\} = 0\]
\noindent
where $F_n(x) = \alpha_n \int_{E_n} (z-x)\mu_n(x,dz)$.

}\end{theorem}

\paragraph{Proof of correctness. }

We recall, from Section \ref{sec:MeanField}, that restrictions are as usual
when dealing with mean field approximations: the numbers of the agents must
be sufficiently large and process constants must be guarded to avoid
infinite-sized systems. Essentially, this result proves that, assuming that the various conditions hold,
the difference between state changes in the Markov chain and the ones expressed
in the relevant ODE are in fact equivalent.
  \\

\begin{theorem}
Given $\mathit{System} = (\Pi_{1\leq j \leq m}P_1\loc{\s_1,\ell_j,q_{1,j}} |  \ldots |\Pi_{1\leq j\leq m}P_n\loc{\s_n,\ell_j,q_{n,j}})\hide L$ in \spalps{},
the system of recurrence equations with variables $P_1(t)@\ell_1 \dots P_1(t)@\ell_m,\ldots, P_n(t)@\ell_1, \dots P_n(t)@\ell_m$
represents the average behavior of the system at any discrete-time unit $t$.
\end{theorem}

\begin{proof}
We show that the recurrence equations obtained from the mean-field semantics of \spalps{} are
equivalent to the recurrence equations obtained from the  mean-field approximation of
the derivation of \spalps{}'s operational semantics into a \textsc{dtmc}.

To do this we must confirm that the condition of the theorem above are satisfied.
\begin{enumerate}
\item According to the theorem, we require a set of \textsc{dtmc} processes $X_n(k)$.
    Indeed, the states of an \spalps{} system is in $\mathbb{N}^k$, where $k$ is
    the number of different processes in the system times the number of locations.
    This is a consequence of the initial-state matrix defined in Section \ref{sec:MeanField}.
\item When considering processes over $\{0,1,\ldots,n\}$, Kurtz rescales such processes to $[0,1]$
    by dividing each element by $n$ and letting $n \to \infty$. Processes in a system of \spalps{}
    range over $\{0,1,\ldots,n\}$, where $n$ is the number of different processes in the system
    times the number of locations. They can also be rescaled to match the defined conditions.
\item According to the theorem, the one step transition function is represented by
$\mu_x(x,\Gamma) = P\{X_n(k+1) \in \Gamma | X_n(k) = x \}$.
If we consider the labelled-transition system of \spalps{} under abstract bisimulation,
the one step transition function can be extracted as
$\mu_n(P, \{P' | P \byp{\beta[p]} P'\}) = p$.
\item Finally, the theorem assumes that there exists sequence of positive numbers $\alpha_n$ and $\epsilon_n$ such that
$\lim\limits_{x\to\infty} \alpha_n = \infty$ and $\lim_{x\to\infty} \epsilon_n = 0$, and
\[\sup_{n} \sup\limits_{x \in  E_n} \alpha_n \int\limits_{E_n} |z-x| \mu_n (x,dz) < \infty\]
and
\[\lim\limits_{x\to\infty} \sup\limits_{x \in E_n} \alpha_n \int\limits_{|z-x| > \epsilon_n} |z-x| \mu_n (x,dz) = 0\]

In terms of \spalps{}, we can think of $x$ and $z$ being state matrices
with a component representing each type of process in the system at each
location. The term $|z-x|$ denotes the difference between the two states $x$ and $z$.
As $n \rightarrow \infty$, the number of states that can be
reached in one step becomes very large. Furthermore, there is higher probability of
moving to a state with  small change from the previous state when there are
lots of components since there are lots of ways to make that change.
Similarly, the states for which the change is high are less likely to occur.
In addition we have that, since process are rescaled to $[0,1]$, then
$0\leq |x-z| \leq 1$, and since $\mu(x,z)$ is a probability, we conclude
that $0\leq \mu(x,z)\leq 1$. Given the relation of probabilities and degree of
changes noted above, we conclude that $\int\limits_{E_n} |z-x| \mu_n (x,dz)\to 0$
 and that  as $n$ and $\a_n$  tend to infinity,
$\alpha_n \int\limits_{E_n} |z-x| \mu_n (x,dz) < \infty$ as required. Similarly,
we may verify the last relation by noting that it captures the states
for which $\mu(x,z) = 0$.

Thus, by Theorem \ref{KurtzTheorem}, for every $\delta>0$, $t>0$, we have
\[\lim\limits_{n \to \infty} \sup\limits_{x \in E_n} P \{ \sup\limits_{k \leq \alpha_n t} | X_n(k) - X_n(0) - \sum_{l=0}^k \frac{1}{\alpha_n} F_n(X_n(l)) | > \delta \mbox{ where } X_n(0) = x\} = 0\]
\noindent
where $F_n(x) = \alpha_n \int_{E_n} (z-x)\mu_n(x,dz)$.
Applying the above over a single time step we obtain
\[\lim\limits_{n \to \infty} \sup\limits_{x \in E_n} P \{ \sup\limits_{k \leq \alpha_n t} | X_n(1) - X_n(0) - \int_{E_n} (z-x)\mu_n(x,dz) | > \delta \mbox{ where } X_n(0) = x\} = 0.\]
This implies that as $n\to \infty$, the difference $ X_n(1) - X_n(0) - \int_{E_n} (z-x)\mu_n(x,dz)$ tends to $0$, therefore
 $X_n(1) = X_n(0) + \int_{E_n} (z-x)\mu_n(x,dz)$, and, since Markov process are memoryless, we conclude that
 \[ X_n(1) = X_n(0) + \int_{E_n} (z-x)\mu_n(x,dz)\]
So, finally, by noting that $\int_{E_n} (z-x)\mu_n(x,dz)$ is equivalent to the way the mean-field equations of Section~\ref{sec:MeanField} are constructed, the result follows.

\end{enumerate}
\end{proof}

\section{Case study: Population dynamics of dengue in Bello, Colombia}
\label{sec:case-study}

An interesting case study where system components can number in millions
is the eco-epidemiology of dengue. Dengue is a disease caused by a virus 
transmitted to humans by the bite of the \emph{Aedes aegypti} mosquito. 
To date, there is no available treatment nor specific vaccine for this disease.
Dengue is a serious public health problem in Colombia. During the last $10$ years, 
there were around $600,000$ cases, from which $9\%$ corresponded to aggravated 
forms of the disease \cite{PAHO2011}. Unfortunately, current programs to prevent
and control dengue in Colombia are insufficient \cite{PAHO2010}. In the Valley 
of Aburr\'{a} (Department of Antioquia), the city of Bello is one of the most 
affected by dengue. In Bello, dengue is endemic; the rate oscillated from $11.1$ to 
$427$ cases by $100,000$ inhabitants, during the years 2002-2009.

Given the endemic status of dengue in Bello, it is important to analyze the factors 
involved in the eco-epidemiology of the dengue disease. Previous results have shown the 
influence of environmental variables in the distributions of cases of the 
disease~\cite{Arboleda2011}. Ongoing work carried out by three of the authors,
Arboleda, Puerta and V\'{e}lez, aims to analyze the macro and micro climatic and 
population factors to determine cases of dengue in Bello\footnote{Research 
carried out within a project founded by Colombian research agency Colciencias: 
``Design and Computational Implementation of a Mathematical Model for the Prediction
 of Occurrences of Dengue''}. There is a disadvantage with such models: the models
are population models based on differential equations; it means, that they
analyze the average behavior of the populations, but it is not possible to know 
how low-level specifications at the individual level will affect the population 
behavior.

The model we present is an individual-based version of the model 
presented in \cite{Yang2008}. To establish the initial conditions 
for the model defined with respect to the human population, we adopted
a total population size of 403,235, as recorded for the urban area of Bello (Antioquia) in
2010 by the \emph{Colombian administrative department of statistics}
\footnote{http://www.dane.gov.co}. The size of the susceptible human 
population at the beginning of the last registered  epidemic was
estimated based on the risk map developed by Arboleda et al.~\cite{Arboleda2009},
in which the probability of infection was reported to be 0.3 in 2008 and 2009, 
with a standard deviation of 0.096; thus, it was
determined that the size of the susceptible human population should be between
244,402 and 321,734.
The initial condition considered for the infectious human population was the number of cases reported
at the beginning of the epidemic. There were six reported cases in Bello during the first week
of the epidemic (week 51 in 2009). Because of under-reporting concerns, which can affect up to 75\% of
the total number of cases \cite{burden2013}, we assumed that the initial number of infectious human individuals should
lie in the range of 6 to 24. In what follows we explain how we model the mosquitoes and the humans. 

\paragraph{Mosquitoes. } The aquatic phases of the mosquito's life cycle are described briefly. The egg,
larva and pupa states are represented by $E$, $L$ and $P$, respectively. Parameters $\sigma_e, \sigma_l, \sigma_p$  are the probabilities to change from egg to larva, from larva to pupa and from pupa to adult. Co-action  $\ol{\mathit{infect}}$ models when a mosquito infects a human. We only consider female mosquitoes in the model. When a mosquito infects a human, it may reproduce and produce 3 offspring; otherwise, it will die in the next time unit.
  Mosquitoes can migrate from one
district to another. 

\paragraph{Humans. } The dynamics of dengue transmission in human population is described by susceptible $(s)$, exposed $(e)$, infectious $(i)$, recovered $(r)$ and dead $(d)$ individuals. Parameter $\mu_h$ represents the probability of an infected human to die from the disease. Action  $\mathit{infect}$ represents the action of being infected by a mosquito. When a human being is infected by a mosquito, it will remain in the exposed state for 3 time units; afterwards, it will become infectious. For simplicity, we do not consider how humans may infect mosquitoes when the mosquitoes bite on infected humans and we do not consider how humans migrate among districts.

{\footnotesize
\begin{center}
\begin{tabular}{l l l l l l l}
$E$ &$\eqdef$ &$\; \sigma_e{:}\surd.L\; \oplus \;(1-\sigma_e){:}\surd.W_6 $& \hspace{0.6in}&$s$ &$\eqdef $&$\; \ol{\mathit{infect}}?(\surd . e, \surd . s) $\\
$L$ &$\eqdef$ &$\; \sigma_l{:}\surd.P\; \oplus \;(1-\sigma_l){:}\surd.W_6 $& &$ e $&$\eqdef $&$\; \surd. e_1$\\
$P$ &$\eqdef$ &$\; \sigma_p{:}\surd.W\; \oplus \;(1-\sigma_p){:}\surd.W_6 $& &$ e_1$ &$\eqdef$ &$\; \surd. e_2$\\
$W$ &$\eqdef$ &$ \mathit{infect}?\, ( W_3,\;\surd. W_1)$ & & $e_2 $&$\eqdef$ &$\; \surd. i$ \\
$W_1$ &$\eqdef$ &$  \sum _{\ell \in \Nb(myloc)} \frac{1}{\| \Nb(myloc) \|}:go\,\ell.W_4$ & &$ i$ &$\eqdef$ &$\; \mu_h{:}\surd . d\; \oplus \;(1-\mu_h){:}\surd . r $\\
$W_2$ &$\eqdef$ &$ \mathit{infect}?\, (W_3,\;\surd. W_5)$ & & $r $&$\eqdef$ &$ \surd. r$\\
$W_3$ &$\eqdef$ &$ \surd. W |E |E |E$ & & $d $& $\eqdef$ &$ \surd. d$\\
$W_4$ &$\eqdef$ &$ \surd.W2$ \\
$W_5$ &$\eqdef$ &$ \nil$

\end{tabular}
\end{center}
}







\paragraph{Mean-field semantics. }
In this case study, we have a total of $16$ process states: $E$, $L$, $P$, $W$, $W_1$ ... $W_5$, $s$, $e$, $e_1$, $e_2$, $i$, $r$, $d$, but
only $s, e, i, r, W, W_2$ are of interest. Furthermore, there are $11$ districts in Bello, Colombia, thus we have $11$ locations.

\begin{enumerate}

\item The initial state is represented by a $11 \times 16$ matrix. 

 
\item The state-transition table for the example is given below.
Note that in fact this is a $3$-dimensional matrix, the third dimension
being the location dimension. To capture this in two dimensions we write $@\ell'$ whenever
the resulting individuals have moved to another location, $\ell'$ being this location.

First, we present the state transitions for the mosquitoes.

{\scriptsize
\begin{center}
\begin{tabular}{|c | c | c | c | c | c | c | c|c|       c| c|}
\hline
&  & $E$ & $L$ & $P$ & $W$ & $W_1$ & $W_2$ & $W_3$     & $W_4$ & $W_5$ \\
\hline
$E\loc{\s,\ell,q}$ & $prob$ &  &$\sigma_e \cdot q$ &  & & & &         & & $(1 - \sigma_e) \cdot q$\\
\hline
$L\loc{\s,\ell,q}$ & $prob$ & & & $\sigma_l \cdot q$ &   & & &                  & & $(1 - \sigma_l) \cdot q$\\
\hline
$P\loc{\s,\ell,q}$ & $prob$ & & & & $\sigma_p \cdot q$  & & &                 & & $(1 - \sigma_p) \cdot q$\\
\hline
$W\loc{\s,\ell,q_W}$ & $infect$ & & & & & $q_W -m$ & &  $m$                & & \\
\hline
$W_1\loc{\s,\ell,q}$ & $prob$ & & & & & & &                 & $z$ & \\
\hline
$W_2\loc{\s,\ell,q_{W2}}$ & $infect $ &  & & & & & &         $n$         & & $q_{W2} - n$\\
\hline
$W_3\loc{\s,\ell,q}$ & $\surd $ &  $3 \cdot q$ & & & $q$ & & &                  & & \\
\hline
$W_4\loc{\s,\ell,q}$ & $\surd $ & & & & & & $q$ &                  & & \\
\hline
$W_5\loc{\s,\ell,q}$ & $\surd $ &  & & &  & & &                  & & $q$\\
\hline
\end{tabular}
\end{center}
}

Second, we present the state transtitions for the humans.

\begin{center}
\begin{tabular}{|c | c | c | c | c | c | c | c|c|}
\hline
&  & $s$ & $e$ & $e_1$ & $e_2$ & $i$ & $r$ & $d$ \\
\hline
$s\loc{\s,\ell,q_s}$ & $\ol{infect}$ & $p $ &$q_s - p$ &  & & & & \\
\hline
$e\loc{\s,\ell,q}$ & $\surd$ & & & $q$ &   & & & \\
\hline
$e_1\loc{\s,\ell,q}$ & $\surd$ & & & & $q$  & & & \\
\hline
$e_2\loc{\s,\ell,q}$ & $\surd$ & & & & & $q$ & &  \\
\hline
$i\loc{\s,\ell,q}$ & $prob$ & & & & & & $(1-\mu_h)\cdot q$ & $\mu_h \cdot q$ \\
\hline
$r\loc{\s,\ell,q}$ & $\surd $ &  & & & & & $q$ &  \\
\hline
$d\loc{\s,\ell,q}$ & $\surd $ &  & & & & & & $q$ \\
\hline
\end{tabular}
\end{center}

where $m =  min(q_w, \frac{q_s}{q_W+q_{W2}})$, $n = min(q_{W2}, \frac{q_s}{q_{W}+q_{W2}}) $, $p = min(q_s, q_W + q_{W_2})$ and \\ $z = \sum_{\ell' \in \Nb(\ell)} \frac{1}{\| \Nb(\ell') \|} \cdot q@\ell'$. Variable $z$ represents the mosquito's dispersal.

\item Let us write $P(t)@\ell$ for the average number of individuals of
 process $P$, at time $t$ and location $\ell$.
 In what follows, we describe $s(t)@\ell $,  $e(t)@\ell $ ... $W_2(t)@\ell$, for $t>0$, which represents the mean-field equations for the behavior of the mosquitoes and humans, respectively. 
 Note that, for each process of interest, we have $11$ equations, one for each location that represents each district in Bello. 
A detailed explanation on how we computed the mean-field semantics of the case study is in~\cite{mean-field-techreport}. 





By manipulating the equations and restricting attention to how the system
evolves between $\surd$ actions and the processes of interest, we obtain:

 \begin{eqnarray*}
 W(t)@\ell & = & \sigma_P \cdot \sigma_L \cdot \sigma_E \cdot 3 \cdot \Bigg( min \bigg( W(t-4)@\ell, \frac{W(t-4)@\ell \cdot s(t-4)@\ell}{W(t-4)@\ell + W_2(t-4)@\ell}\bigg)  \\
 & & + min \bigg( W_2(t-4)@\ell, \frac{W_2(t-4)@\ell \cdot s(t-4)@\ell)}{W(t-4)@\ell+W_2(t-4)@\ell} \bigg)      \Bigg)\\
 W_2(t)@\ell & = & \frac{1}{ \| \Nb(\ell') \|  } \cdot \Bigg(   W_1(t-3)@\ell' \\
 & & -  min \bigg(  W(t-3)@\ell', \frac{W(t-3)@\ell' \cdot s(t-3)@\ell'} {W(t-3)@\ell' + W_2(t-3)@\ell' }  \bigg)      \Bigg)\\
 s(t)@\ell & = & min(s(t-1)@\ell, W(t-1)@\ell + W_2(t-1)@\ell\\
 e(t)@\ell & = & e(t-1)@\ell - min(s(t-1)@\ell, W(t-1)@\ell + W_2(t-1)@\ell\\
 i(t)@\ell & = & e(t-3)@\ell\\
 r(t)@\ell & = & r(t-1)@\ell + (1-\mu_h) \cdot i(t-1)@\ell
\end{eqnarray*}

where $\ell' \in \Nb(\ell)$

\end{enumerate}

Up to our knowledge, this is the first system of recurrence equations developed for dengue. We leave as future work
the validation of the model with real data.

\section{Conclusions}
\label{sec:Conclusions}


In this paper we presented a mean-field semantics for \spalps{}. Up to our knowledge, 
\spalps{} is the first spatially-explicit
probabilistic process calculus to be extended with mean-field semantics. Using this semantics
we can analyze deterministically the average behavior of a spatially-explicit ecological
model even for large populations. The advantages of this new semantics is that 
it allows us to translate from an individual-based model to the underlying population
dynamics and that it is possible to do this efficiently without computing the complete
state space of the model.  In particular, we showed how this is applicable for
epidemiological models by our case study on the transmission of dengue. 



As future work, we want to further study the spatial distribution of dengue on the lines of Otero et al. \cite{Otero2010}. In fact, there is demographic and epidemiological information about the
reported cases in each district of Bello (Antioquia), Colombia. It is of vital importance for public health to determine which districts have more risk by determining the migration patterns of mosquitoes from district to district. This is of importance to define vaccination and fumigation schemes to prevent epidemics.

\let\oldbibliography\thebibliography
\renewcommand{\thebibliography}[1]{%
  \oldbibliography{#1}%
  \setlength{\itemsep}{0pt}
}

\bibliographystyle{eptcs}
{\footnotesize
\bibliography{biblio2015}

\begin{thebibliography}{10}
\providecommand{\bibitemdeclare}[2]{}
\providecommand{\surnamestart}{}
\providecommand{\surnameend}{}
\providecommand{\urlprefix}{Available at }
\providecommand{\url}[1]{\texttt{#1}}
\providecommand{\href}[2]{\texttt{#2}}
\providecommand{\urlalt}[2]{\href{#1}{#2}}
\providecommand{\doi}[1]{doi:\urlalt{http://dx.doi.org/#1}{#1}}
\providecommand{\bibinfo}[2]{#2}

\bibitemdeclare{misc}{PAHO2010}
\bibitem{PAHO2010}
\emph{\bibinfo{title}{Gu{\'i}a de Atenci{\'o}n Cl{\'i}nica Integral del
  Paciente con Dengue}}.
\newblock
  \urlprefix\url{http://calisaludable.cali.gov.co/saludPublica/2010\_Dengue/protocolos\_guias\_del\_INS\_y\_del\_MPS/Guia\_Dengue\_2010.pdf}.

\bibitemdeclare{misc}{PAHO2011}
\bibitem{PAHO2011}
\emph{\bibinfo{title}{Number of Reported Cases of Dengue and Dengue Hemorrhagic
  Fever (DHF) in the Americas, by Country}}.
\newblock
  \urlprefix\url{http://new.paho.org/hq/index.php?option=com\_content\&task=blogcategory\&id=1221\&Itemid=2481}.

\bibitemdeclare{misc}{prism}
\bibitem{prism}
\emph{\bibinfo{title}{Online PRISM documentation}}.
\newblock \urlprefix\url{http://www.prismmodelchecker.org/doc/}.

\bibitemdeclare{article}{Arboleda2011}
\bibitem{Arboleda2011}
\bibinfo{author}{S.~\surnamestart Arboleda\surnameend}, \bibinfo{author}{O.N.
  \surnamestart Jaramillo\surnameend} \& \bibinfo{author}{A.T. \surnamestart
  Peterson\surnameend} (\bibinfo{year}{2012}): \emph{\bibinfo{title}{{ Spatial
  and temporal dynamics of Aedes aegypti breeding sites in Bello, Colombia}}}.
\newblock {\sl \bibinfo{journal}{Journal of Vector Ecology}}
  \bibinfo{volume}{37}(\bibinfo{number}{1}), pp. \bibinfo{pages}{37--48},
  \doi{10.1111/j.1948-7134.2012.00198.x}.

\bibitemdeclare{article}{Arboleda2009}
\bibitem{Arboleda2009}
\bibinfo{author}{Sair \surnamestart Arboleda\surnameend},
  \bibinfo{author}{Nicolas \surnamestart Jaramillo~O.\surnameend} \&
  \bibinfo{author}{A.~Townsend \surnamestart Peterson\surnameend}
  (\bibinfo{year}{2009}): \emph{\bibinfo{title}{Mapping Environmental
  Dimensions of Dengue Fever Transmission Risk in the {A}burr\'{a} {V}alley,
  {C}olombia}}.
\newblock {\sl \bibinfo{journal}{International Journal of Environmental
  Research and Public Health}} \bibinfo{volume}{6}(\bibinfo{number}{12}), p.
  \bibinfo{pages}{3040}, \doi{10.3390/ijerph6123040}.

\bibitemdeclare{article}{BMMP09}
\bibitem{BMMP09}
\bibinfo{author}{R.~\surnamestart Barbuti\surnameend},
  \bibinfo{author}{A.~\surnamestart Maggiolo-Schettini\surnameend},
  \bibinfo{author}{P.~\surnamestart Milazzo\surnameend} \&
  \bibinfo{author}{G.~\surnamestart Pardini\surnameend} (\bibinfo{year}{2011}):
  \emph{\bibinfo{title}{{Spatial Calculus of Looping Sequences}}}.
\newblock {\sl \bibinfo{journal}{Theoretical Computer Science}}
  \bibinfo{volume}{412}(\bibinfo{number}{43}), pp. \bibinfo{pages}{5976--6001},
  \doi{10.1016/j.tcs.2011.01.020}.

\bibitemdeclare{article}{BarbutiMMT06}
\bibitem{BarbutiMMT06}
\bibinfo{author}{Roberto \surnamestart Barbuti\surnameend},
  \bibinfo{author}{Andrea \surnamestart Maggiolo-Schettini\surnameend},
  \bibinfo{author}{Paolo \surnamestart Milazzo\surnameend} \&
  \bibinfo{author}{Angelo \surnamestart Troina\surnameend}
  (\bibinfo{year}{2006}): \emph{\bibinfo{title}{{A Calculus of Looping
  Sequences for Modelling Microbiological Systems}}}.
\newblock {\sl \bibinfo{journal}{Fundamenta Informaticae}}
  \bibinfo{volume}{72}(\bibinfo{number}{1-3}), pp. \bibinfo{pages}{21--35}.

\bibitemdeclare{inproceedings}{abs-1008-3301}
\bibitem{abs-1008-3301}
\bibinfo{author}{Thomas~Anung \surnamestart Basuki\surnameend},
  \bibinfo{author}{Antonio \surnamestart Cerone\surnameend},
  \bibinfo{author}{Roberto \surnamestart Barbuti\surnameend},
  \bibinfo{author}{Andrea \surnamestart Maggiolo-Schettini\surnameend},
  \bibinfo{author}{Paolo \surnamestart Milazzo\surnameend} \&
  \bibinfo{author}{Elisabetta \surnamestart Rossi\surnameend}
  (\bibinfo{year}{2010}): \emph{\bibinfo{title}{{Modelling the Dynamics of an
  Aedes albopictus Population}}}.
\newblock In: {\sl \bibinfo{booktitle}{{Proceedings of AMCA-POP'10}}}, {\sl
  \bibinfo{series}{{EPTCS}}}~\bibinfo{volume}{33}, pp. \bibinfo{pages}{18--36},
  \doi{10.4204/EPTCS.33.2}.

\bibitemdeclare{article}{BCPM06}
\bibitem{BCPM06}
\bibinfo{author}{D.~\surnamestart Besozzi\surnameend},
  \bibinfo{author}{P.~\surnamestart Cazzaniga\surnameend},
  \bibinfo{author}{D.~\surnamestart Pescini\surnameend} \&
  \bibinfo{author}{G.~\surnamestart Mauri\surnameend} (\bibinfo{year}{2008}):
  \emph{\bibinfo{title}{{Modelling metapopulations with stochastic membrane
  systems}}}.
\newblock {\sl \bibinfo{journal}{BioSystems}}
  \bibinfo{volume}{91}(\bibinfo{number}{3}), pp. \bibinfo{pages}{499--514},
  \doi{10.1016/j.biosystems.2006.12.011}.

\bibitemdeclare{article}{burden2013}
\bibitem{burden2013}
\bibinfo{author}{S.~\surnamestart Bratt\surnameend}, \bibinfo{author}{P.~W.
  \surnamestart Gething\surnameend}, \bibinfo{author}{O.~J. \surnamestart
  Brady\surnameend}, \bibinfo{author}{J.~P. \surnamestart Messina\surnameend},
  \bibinfo{author}{A.~W. \surnamestart Farlow\surnameend},
  \bibinfo{author}{C.~L. \surnamestart Moyes\surnameend},
  \bibinfo{author}{J.~M. \surnamestart Drake\surnameend},
  \bibinfo{author}{J.~S. \surnamestart Brownstein\surnameend},
  \bibinfo{author}{A.~G. \surnamestart Hoen\surnameend}, \bibinfo{author}{Osman
  \surnamestart Sankoh\surnameend}, \bibinfo{author}{Monica~F. \surnamestart
  Myers\surnameend}, \bibinfo{author}{Dylan~B. \surnamestart
  George\surnameend}, \bibinfo{author}{Thomas \surnamestart
  Jaenisch\surnameend}, \bibinfo{author}{G.~R.~William \surnamestart
  Wint\surnameend}, \bibinfo{author}{Cameron~P. \surnamestart
  Simmons\surnameend}, \bibinfo{author}{Thomas~W. \surnamestart
  Scott\surnameend}, \bibinfo{author}{Jeremy~J. \surnamestart
  Farrar\surnameend} \& \bibinfo{author}{Simon~I. \surnamestart Hay\surnameend}
  (\bibinfo{year}{2013}): \emph{\bibinfo{title}{{The global distribution and
  burden of dengue}}}.
\newblock {\sl \bibinfo{journal}{{Nature}}} \bibinfo{volume}{496}, pp.
  \bibinfo{pages}{504--507}, \doi{10.1038/nature12060}.

\bibitemdeclare{article}{CYL09}
\bibitem{CYL09}
\bibinfo{author}{Qiuwen \surnamestart Chen\surnameend}, \bibinfo{author}{Fei
  \surnamestart Ye\surnameend} \& \bibinfo{author}{Weifeng \surnamestart
  Li\surnameend} (\bibinfo{year}{2009}):
  \emph{\bibinfo{title}{{Cellular-automata-based ecological and ecohydraulics
  modelling}}}.
\newblock {\sl \bibinfo{journal}{Journal of Hydroinformatics}}
  \bibinfo{volume}{11}(\bibinfo{number}{3/4}), pp. \bibinfo{pages}{252--272},
  \doi{10.2166/hydro.2009.026}.

\bibitemdeclare{article}{reactiveNetworks}
\bibitem{reactiveNetworks}
\bibinfo{author}{Fr{\'{e}}d{\'{e}}ric \surnamestart Didier\surnameend},
  \bibinfo{author}{Thomas~A. \surnamestart Henzinger\surnameend},
  \bibinfo{author}{Maria \surnamestart Mateescu\surnameend} \&
  \bibinfo{author}{Verena \surnamestart Wolf\surnameend}
  (\bibinfo{year}{2010}): \emph{\bibinfo{title}{{SABRE:} {A} Tool for
  Stochastic Analysis of Biochemical Reaction Networks}}.
\newblock {\sl \bibinfo{journal}{CoRR}} \bibinfo{volume}{abs/1005.2819},
  \doi{10.1109/QEST.2010.33}.

\bibitemdeclare{article}{DrabikMM11}
\bibitem{DrabikMM11}
\bibinfo{author}{Peter \surnamestart Dr{\'a}bik\surnameend},
  \bibinfo{author}{Andrea \surnamestart Maggiolo-Schettini\surnameend} \&
  \bibinfo{author}{Paolo \surnamestart Milazzo\surnameend}
  (\bibinfo{year}{2011}): \emph{\bibinfo{title}{{Modular Verification of
  Interactive Systems with an Application to Biology}}}.
\newblock {\sl \bibinfo{journal}{Scientific Annals of Computer Science}}
  \bibinfo{volume}{21}(\bibinfo{number}{1}), pp. \bibinfo{pages}{39--72},
  \doi{10.1016/j.entcs.2010.12.006}.

\bibitemdeclare{inproceedings}{FM04}
\bibitem{FM04}
\bibinfo{author}{S.~C. \surnamestart Fu\surnameend} \&
  \bibinfo{author}{G.~\surnamestart Milne\surnameend} (\bibinfo{year}{2004}):
  \emph{\bibinfo{title}{{A Flexible Automata Model for Disease Simulation}}}.
\newblock In: {\sl \bibinfo{booktitle}{{Proceedings of ACRI'04}}},
  \bibinfo{series}{{LNCS 3305}}, \bibinfo{publisher}{Springer}, pp.
  \bibinfo{pages}{642--649}, \doi{10.1007/978-3-540-30479-1_66}.

\bibitemdeclare{book}{Graham1994}
\bibitem{Graham1994}
\bibinfo{author}{Ronald~L. \surnamestart Graham\surnameend},
  \bibinfo{author}{Donald~E. \surnamestart Knuth\surnameend} \&
  \bibinfo{author}{Oren \surnamestart Patashnik\surnameend}
  (\bibinfo{year}{1994}): \emph{\bibinfo{title}{Concrete Mathematics: A
  Foundation for Computer Science}}, \bibinfo{edition}{2nd} edition.
\newblock \bibinfo{publisher}{Addison-Wesley Longman Publishing Co., Inc.},
  \bibinfo{address}{Boston, MA, USA}, \doi{10.1063/1.4822863}.

\bibitemdeclare{article}{HillstonTG12}
\bibitem{HillstonTG12}
\bibinfo{author}{Jane \surnamestart Hillston\surnameend},
  \bibinfo{author}{Mirco \surnamestart Tribastone\surnameend} \&
  \bibinfo{author}{Stephen \surnamestart Gilmore\surnameend}
  (\bibinfo{year}{2012}): \emph{\bibinfo{title}{{Stochastic Process Algebras:
  From Individuals to Populations}}}.
\newblock {\sl \bibinfo{journal}{Computing Journal}}
  \bibinfo{volume}{55}(\bibinfo{number}{7}), pp. \bibinfo{pages}{866--881},
  \doi{10.1093/comjnl/bxr094}.

\bibitemdeclare{article}{KleijnKR11}
\bibitem{KleijnKR11}
\bibinfo{author}{Jetty \surnamestart Kleijn\surnameend},
  \bibinfo{author}{Maciej \surnamestart Koutny\surnameend} \&
  \bibinfo{author}{Grzegorz \surnamestart Rozenberg\surnameend}
  (\bibinfo{year}{2011}): \emph{\bibinfo{title}{{Petri Nets for Biologically
  Motivated Computing}}}.
\newblock {\sl \bibinfo{journal}{Scientific Annals of Computer Science}}
  \bibinfo{volume}{21}(\bibinfo{number}{2}), pp. \bibinfo{pages}{199--225}.

\bibitemdeclare{article}{Kurtz1970}
\bibitem{Kurtz1970}
\bibinfo{author}{Thomas \surnamestart Kurtz\surnameend} (\bibinfo{year}{1970}):
  \emph{\bibinfo{title}{{Solutions of Ordinary Differential Equations as Limits
  of Pure Jump Markov Processes}}}.
\newblock {\sl \bibinfo{journal}{Journal of Applied Probability}}
  \bibinfo{volume}{7}(\bibinfo{number}{1}), \doi{10.2307/3212147}.

\bibitemdeclare{inproceedings}{MNS08}
\bibitem{MNS08}
\bibinfo{author}{C.~\surnamestart McCaig\surnameend},
  \bibinfo{author}{R.~\surnamestart Norman\surnameend} \&
  \bibinfo{author}{C.~\surnamestart Shankland\surnameend}
  (\bibinfo{year}{2008}): \emph{\bibinfo{title}{{Process Algebra Models of
  Population Dynamics}}}.
\newblock In: {\sl \bibinfo{booktitle}{{Proceedings of AB'08}}},
  \bibinfo{series}{{LNCS 5147}}, \bibinfo{publisher}{Springer}, pp.
  \bibinfo{pages}{139--155}, \doi{10.1007/978-3-540-85101-1_11}.

\bibitemdeclare{article}{MNS10}
\bibitem{MNS10}
\bibinfo{author}{Chris \surnamestart McCaig\surnameend},
  \bibinfo{author}{Rachel \surnamestart Norman\surnameend} \&
  \bibinfo{author}{Carron \surnamestart Shankland\surnameend}
  (\bibinfo{year}{2011}): \emph{\bibinfo{title}{{From individuals to
  populations: A mean field semantics for process algebra}}}.
\newblock {\sl \bibinfo{journal}{Theoretical Computer Science}}
  \bibinfo{volume}{412}(\bibinfo{number}{17}), pp. \bibinfo{pages}{1557--1580},
  \doi{10.1016/j.tcs.2010.09.024}.

\bibitemdeclare{article}{Otero2010}
\bibitem{Otero2010}
\bibinfo{author}{M.~Solari~H.G. \surnamestart Otero\surnameend}
  (\bibinfo{year}{2010}): \emph{\bibinfo{title}{{Stochastic eco-epidemiological
  model of dengue disease transmission by Aedes aegypti mosquito}}}.
\newblock {\sl \bibinfo{journal}{Mathematical Biosciences}}
  \bibinfo{volume}{223}(\bibinfo{number}{1}), pp. \bibinfo{pages}{32--46},
  \doi{10.1016/j.mbs.2009.10.005}.

\bibitemdeclare{inproceedings}{pntcc}
\bibitem{pntcc}
\bibinfo{author}{Jorge \surnamestart P{\'e}rez\surnameend} \&
  \bibinfo{author}{Camilo \surnamestart Rueda\surnameend}
  (\bibinfo{year}{2008}): \emph{\bibinfo{title}{{Non-determinism and
  Probabilities in Timed Concurrent Constraint Programming}}}.
\newblock In: {\sl \bibinfo{booktitle}{{Proceedings of ICLP'08}}}, {\sl
  \bibinfo{series}{{LNCS}}} \bibinfo{volume}{5366}, pp.
  \bibinfo{pages}{677--681}, \doi{10.1007/978-3-540-89982-2_56}.

\bibitemdeclare{inproceedings}{PT13}
\bibitem{PT13}
\bibinfo{author}{Anna \surnamestart Philippou\surnameend} \&
  \bibinfo{author}{Mauricio \surnamestart Toro\surnameend}
  (\bibinfo{year}{2013}): \emph{\bibinfo{title}{{Process Ordering in a Process
  Calculus for Spatially-Explicit Ecological Models.}}}
\newblock In: {\sl \bibinfo{booktitle}{{Proceedings of MOKMASD'13}}},
  \bibinfo{series}{{LNCS 8368}}, \bibinfo{publisher}{Springer}, pp.
  \bibinfo{pages}{345--361}, \doi{10.1007/978-3-319-05032-4_25}.

\bibitemdeclare{article}{PTA13}
\bibitem{PTA13}
\bibinfo{author}{Anna \surnamestart Philippou\surnameend},
  \bibinfo{author}{Mauricio \surnamestart Toro\surnameend} \&
  \bibinfo{author}{Margarita \surnamestart Antonaki\surnameend}
  (\bibinfo{year}{2013}): \emph{\bibinfo{title}{{Simulation and Verification
  for a Process Calculus for Spatially-Explicit Ecological Models}}}.
\newblock {\sl \bibinfo{journal}{Scientific Annals of Computer Science}}
  \bibinfo{volume}{23}(\bibinfo{number}{1}), pp. \bibinfo{pages}{119--167},
  \doi{10.7561/SACS.2013.1.119}.

\bibitemdeclare{article}{PinnaS08}
\bibitem{PinnaS08}
\bibinfo{author}{G.~Michele \surnamestart Pinna\surnameend} \&
  \bibinfo{author}{Andrea \surnamestart Saba\surnameend}
  (\bibinfo{year}{2008}): \emph{\bibinfo{title}{{An Event Based Semantics of P
  Systems}}}.
\newblock {\sl \bibinfo{journal}{Scientific Annals of Computer Science}}
  \bibinfo{volume}{18}, pp. \bibinfo{pages}{99--127}.

\bibitemdeclare{book}{Puterman1994}
\bibitem{Puterman1994}
\bibinfo{author}{Martin~L. \surnamestart Puterman\surnameend}
  (\bibinfo{year}{1994}): \emph{\bibinfo{title}{Markov Decision Processes:
  Discrete Stochastic Dynamic Programming}}, \bibinfo{edition}{1st} edition.
\newblock \bibinfo{publisher}{John Wiley \& Sons, Inc.}, \bibinfo{address}{New
  York, NY, USA}, \doi{10.1002/9780470316887}.

\bibitemdeclare{incollection}{timedpi}
\bibitem{timedpi}
\bibinfo{author}{Neda \surnamestart Saeedloei\surnameend} \&
  \bibinfo{author}{Gopal \surnamestart Gupta\surnameend}
  (\bibinfo{year}{2014}): \emph{\bibinfo{title}{Timed PI Calculus}}.
\newblock In \bibinfo{editor}{Martín \surnamestart Abadi\surnameend} \&
  \bibinfo{editor}{Alberto \surnamestart Lluch~Lafuente\surnameend}, editors:
  {\sl \bibinfo{booktitle}{Trustworthy Global Computing}}, {\sl
  \bibinfo{series}{Lecture Notes in Computer Science}} \bibinfo{volume}{8358},
  \bibinfo{publisher}{Springer International Publishing}, pp.
  \bibinfo{pages}{119--135}, \doi{10.1007/978-3-319-05119-2_8}.

\bibitemdeclare{article}{SBB01}
\bibitem{SBB01}
\bibinfo{author}{D.~J.~T. \surnamestart Sumpter\surnameend},
  \bibinfo{author}{G.~B. \surnamestart Blanchard\surnameend} \&
  \bibinfo{author}{D.~S. \surnamestart Broomhear\surnameend}
  (\bibinfo{year}{2001}): \emph{\bibinfo{title}{{Ants and Agents: a Process
  Algebra Approach to Modelling Ant Colony Behaviour}}}.
\newblock {\sl \bibinfo{journal}{Bulletin of Mathematical Biology}}
  \bibinfo{volume}{63}, pp. \bibinfo{pages}{951--980},
  \doi{10.1006/bulm.2001.0252}.

\bibitemdeclare{article}{Tofts94}
\bibitem{Tofts94}
\bibinfo{author}{Chris \surnamestart Tofts\surnameend} (\bibinfo{year}{1994}):
  \emph{\bibinfo{title}{{Processes with probabilities, priority and time}}}.
\newblock {\sl \bibinfo{journal}{Formal Aspects of Computing}}
  \bibinfo{volume}{6}(\bibinfo{number}{5}), p. \bibinfo{pages}{536–564},
  \doi{10.1007/BF01211867}.

\bibitemdeclare{techreport}{mean-field-techreport}
\bibitem{mean-field-techreport}
\bibinfo{author}{Mauricio \surnamestart Toro\surnameend}, \bibinfo{author}{Anna
  \surnamestart Philippou\surnameend}, \bibinfo{author}{Sair \surnamestart
  Arboleda\surnameend}, \bibinfo{author}{Carlos \surnamestart
  V\'{e}lez\surnameend} \& \bibinfo{author}{Mar\'{i}a \surnamestart
  Puerta\surnameend} (\bibinfo{year}{2015}): \emph{\bibinfo{title}{{Mean-field
  semantics for a Process Calculus for Spatially-Explicit Ecological Models}}}.
\newblock \bibinfo{type}{Technical Report}, \bibinfo{institution}{Department of
  Informatics and Systems, Universidad Eafit}.
\newblock \bibinfo{note}{Available at
  {http://blogs.eafit.edu.co/giditic-software/2015/10/01/mean-field/}}.

\bibitemdeclare{inproceedings}{TPSK14}
\bibitem{TPSK14}
\bibinfo{author}{Mauricio \surnamestart Toro\surnameend}, \bibinfo{author}{Anna
  \surnamestart Philippou\surnameend}, \bibinfo{author}{Christina \surnamestart
  Kassara\surnameend} \& \bibinfo{author}{Spyros \surnamestart
  Sfenthourakis\surnameend} (\bibinfo{year}{2014}):
  \emph{\bibinfo{title}{Synchronous Parallel Composition in a Process Calculus
  for Ecological Models}}.
\newblock In: {\sl \bibinfo{booktitle}{Proceedings of ICTAC'14}}, pp.
  \bibinfo{pages}{424--441}, \doi{10.1007/978-3-319-10882-7_25}.

\bibitemdeclare{article}{Tribastone12}
\bibitem{Tribastone12}
\bibinfo{author}{Mirco \surnamestart Tribastone\surnameend},
  \bibinfo{author}{Stephen \surnamestart Gilmore\surnameend} \&
  \bibinfo{author}{Jane \surnamestart Hillston\surnameend}
  (\bibinfo{year}{2012}): \emph{\bibinfo{title}{{Scalable Differential Analysis
  of Process Algebra Models}}}.
\newblock {\sl \bibinfo{journal}{IEEE Transactions on Software Engineering}}
  \bibinfo{volume}{38}(\bibinfo{number}{1}), pp. \bibinfo{pages}{205--219},
  \doi{10.1109/TSE.2010.82}.

\bibitemdeclare{article}{Yang2008}
\bibitem{Yang2008}
\bibinfo{author}{Hyun~Mo \surnamestart Yang\surnameend} \&
  \bibinfo{author}{Cl\'{a}udia~Pio \surnamestart Ferreira\surnameend}
  (\bibinfo{year}{2008}): \emph{\bibinfo{title}{{Assessing the efects of vector
  control on dengue transmission}}}.
\newblock {\sl \bibinfo{journal}{Applied Mathematics and Computation}}
  \bibinfo{volume}{198}, pp. \bibinfo{pages}{401--413},
  \doi{10.1016/j.amc.2007.08.046}.

\end{thebibliography}
}

\end{document}